\documentclass{article}
\usepackage{amssymb,amsmath,latexsym,amscd,amsfonts, bbm}  
\usepackage{graphicx}  
\usepackage[usenames]{color}

\newtheorem{theorem}{Theorem}[section]

\newtheorem{deff}[theorem]{Definition}

\newtheorem{lemma}[theorem]{Lemma}  
\newtheorem{conj}[theorem]{Conjecture}

\newcommand{\qedsymb}{\hfill{\rule{2mm}{2mm}}}  
\newenvironment{proof}[1][]{\begin{trivlist}  
  \item[\hspace{\labelsep}{\bf\noindent Proof#1:\/}] 
    }{\qedsymb\end{trivlist}}

\newcommand{\ignore}[1]{}

\newcommand{\QMA}{\mathsf{QMA}}

\newcommand{\NP}{\mathsf{NP}}

\newcommand{\norm}[1]{{\| #1 \|}}  
  
\newcommand{\ket}[1]{{ |{#1} \rangle }}  
\newcommand{\bra}[1]{{ \langle {#1} | }}
\newcommand{\inprod}[2]{\langle {#1} | {#2}\rangle}

\newcommand{\poly}{\mathrm{poly}} 
\newcommand{\EqDef}{\stackrel{\mathrm{def}}{=}}

\newcommand{\QSAT}{\mathrm{QSAT}}

\newcommand{\Qdit}{\mathbbm{C}^d}
\DeclareMathOperator{\Tr}{Tr}

\newcommand{\Eq}[1]{Eq.~(\ref{#1})}

\newcommand{\Lem}[1]{Lemma~\ref{#1}}

\newcommand{\Ref}[1]{Ref.~\cite{#1}}

\begin{document}

\title{A note about a partial no-go theorem for quantum PCP}
\date{\today}
\author{Itai Arad}
\maketitle 

\begin{abstract}

  This is not a disproof of the quantum PCP conjecture!

  In this note we use perturbation on the commuting Hamiltonian
  problem on a graph, based on results by Bravyi and Vyalyi, to
  provide a partial no-go theorem for quantum PCP. Specifically, we
  derive an upper bound on how large the promise gap can be for the
  quantum PCP still to hold, as a function of the non-commuteness of
  the system. As the system becomes more and more commuting, the
  maximal promise gap shrinks.
  
  We view these results as possibly a preliminary step towards
  disproving the quantum PCP conjecture posed in \cite{ref:Aha09}. A
  different way to view these results is actually as indications
  that a critical point exists, beyond which quantum PCP indeed
  holds; in any case, we hope that these results will lead to
  progress on this important open problem.
\end{abstract}

\section{Introduction}
The PCP theorem is arguably the most important development in
computational complexity over the last two decades. In a nutshell,
what it says is that given a constraint satisfaction problem (CSP),
it can be efficiently replaced by one with comparable size, such that if the
original one was satisfiable, the new one is too, whereas if the
original one had at least one violation in each assignment, then any
assignment to the new CSP must violate a \emph{constant fraction}
(!) of the constraints. Consequently, the problem of deciding
whether a CSP is satisfiable or is violating a constant fraction of
its constraint is NP-hard.

Is there a quantum analogue to this remarkable theorem? This is
perhaps the most important open question in quantum Hamiltonian
complexity, and one of the central problems in quantum complexity in
general. Both a proof and disproof of this conjecture would arguably
yield deep insights into the basic notions of quantum mechanics,
such as entanglement, no-cloning of information, and the quantum to
classical transition on large scales.

In this note we present a result that might be seen as a weak
evidence \emph{against} the existence of a quantum PCP theorem.
Hopefully, it may serve as a starting point for a more general
framework for disproving this conjecture, or alternatively, as a
starting point for better clarifications of the conditions for a
quantum PCP theorem to hold. Before stating the result, we first
state what is meant by a quantum PCP theorem.

\subsection{Background on the quantum PCP conjecture} 

The quantum PCP conjecture was first stated formally in
\Ref{ref:Aha09}. Here we shall roughly follow their presentation.
The quantum analog of a classical CSP is the QSAT problem, which is
a special instance of the local-Hamiltonian problem. In that
problem, we are give a $k$-local Hamiltonian over a system of $n$
qubits $H=\sum_{i=1}^M Q_i$ that is made of $k$-local projections
$Q_i$ with $M=\poly(n)$. We are promised that the ground energy of
the system is either 0 (all quantum constrains are satisfied) or it
is above some constant $a=1/\poly(n)$. Just like its classical
counter-part, this problem is known to be quantum
NP-complete\footnote{It is in fact $\QMA_1$ complete and not $\QMA$
complete, which is the more natural ``quantum-$\NP$'' class, but we
shall not be bothered here with this technical issue. For more
information on this subtle point, see Refs.~\cite{ref:Bra06,
ref:Aha09}.}. 

A quantum PCP theorem would state that even if $a = rM$, for some
constant $0<r<1$, the problem is still quantum-$\NP$ hard to decide. In
other words, it is quantum-$\NP$ hard to distinguish between the case
when the system is completely satisfiable, or when, roughly, a
fraction $r$ of it can not be satisfied. 

Formally, we may define
\begin{deff}[The $r$-gap $k$-$\QSAT$ problem]
\label{def:kr} Let $r\in (0,1)$ be some constant.
  We are given a $k$-local Hamiltonian $H=\sum_{i=1}^M Q_i$
  over $n$ qubits, where the $Q_i$ are $k$-local projections and
  $M=\poly(n)$. We are promised that the ground energy of $H$ is
  either 0 (a YES instance) or is greater than $rM$ (a NO instance).
  We are asked to decide which is which.
\end{deff}
Then the Quantum PCP conjecture is
\begin{conj}[Quantum PCP]
\label{conj:qpcp}
  There exist $(k,r)$ for which Problem~\ref{def:kr} is $\QMA_1$ hard.
\end{conj}

To prove that such a problem is $\QMA_1$ hard, we would like to show
an efficient reduction of another $\QMA_1$-hard problem to it. It is
natural to start with the $k$-QSAT problem that was described above.
We would like to find an efficient transformation that takes a
$k$-QSAT problem and turns it into a $r$-gap $k$-QSAT problem such that if
the original system was satisfiable (ground energy is zero), then so
is the new system. On the other hand, if it was not satisfiable, with
a ground energy above $a=1/\poly$, then the ground energy of
the new system would be above $rM$. 

This type of transformation is called ``Gap Amplification'', because
it amplifies the promise gap of the problem. It is precisely this
type of transformation that was used iteratively in Dinur's proof of
the classical PCP theorem \cite{ref:Din07}. Let us describe this
transformation in some more details. Dinur achieves gap
amplification by an iterative process that amplifies gap by some
constant factor $>1$ at each round. Each such iteration is made of 3
steps. One step amplifies the promise gap of the system at the
expense of making it much less local. The purpose of the two other
steps is to fix this, restoring the locality of the system, without
compromising too much the gap amplification of the first step.

The entire process is carried over a CSP that is defined on an
expander graph, and in addition, the entire proof is very
combinatorial in nature. These two facts make it a promising outline
for a quantum proof under the natural mapping of classical
constraints to projections.  Indeed, a first step in that direction
was taken in \Ref{ref:Aha09}, where it was  shown that essentially the
amplification step in Dinur's proof can be done also quantumly. 
Like the classical proof, the quantum proof relies on
expander graphs. The quantization of the two other steps, however,
remains an open problem.

\subsection{Reasons for doubts in a quantum PCP theorem} 

In the attempts to prove the quantum analogue of Dinur's proof, it
seems hard to quantize any classical step that increases the size of
the system; such steps seem to conflict with the quantum no-cloning
principle. See \Ref{ref:Aha09} for more details; nevertheless, such
increase seems unavoidable in the classical case.

Except for these difficulties, there is another reason to believe
that there is no quantum PCP theorem, pointed to us by Hastings
\cite{ref:Has:nopcp}.  As we have seen, such theorem implies the
existence of systems in which it is quantum-$\NP$ hard to
distinguish between a vanishing ground energy and ground energy of
the order of the system size. From a physicist point of view, it is
equivalent to determining whether or not the free energy of the
system becomes negative at a \emph{finite} temperature \cite{ref:Pou10}. It
is then argued that at such temperatures, on large scale, the system
loses its quantum characteristics; long-range entanglement effects
must fade. Consequently, the system can be described (approximately)
classically, hence the problem is inside NP.

In this note, we will pursue this direction. The ultimate goal is to
show that for any $(k,r)$, Problem~\ref{def:kr} is inside $\NP$. This, 
however, seems very difficult.  Instead of attacking it directly,
it might be beneficial to show first that a more restricted problem
is inside $\NP$. This is what we do here.

\subsection{Results: A partial No-Go theorem for quantum PCP} 
We are interested in a version of
problem (\ref{def:kr}) in which the projections are two-local,
sitting on the edges of a $D$-regular graph:

\begin{deff}[The $(d,D,r)$-gap Hamiltonian problem on a graph]
\label{def:dr} 

  We consider a QSAT system $H=\sum_{i=1}^M Q_i$ that is defined on
  a $D$-regular graph, using $d$-dimensional qudits that sit on its
  vertices, and projections $\{Q_i\}$ that sit on its edges. We are
  promised that the ground energy of $H$ is either 0 or is greater
  than $rM$ for some constant $0<r<1$, and we are asked to decide
  which is which.
\end{deff}

The advantage of working with this restricted set of problems is
three-fold. First, its classical analog is the outcome of Dinur's
classical PCP proof \cite{ref:Din07}. It is therefore a natural
candidate for a quantum PCP construction. Second, it has a
classical, yet non-trivial limit, which was discovered by Bravyi and
Vyalyi \cite{ref:Bra03}: When the projections commute, the problem
becomes classical in the sense that the ground state of the system
can be described by a shallow tensor-network that can be contracted
efficiently on a classical computer. This tensor-network can be
given as a witness to the prover, and hence the problem is inside
$\NP$. Finally, by itself, this class of systems seems rich enough
to capture non-trivial quantum effects, if exist. For example,
with a $1/\poly$ promise gap, these systems become $\QMA_1$-hard
\cite{ref:Bra06}.

In this note we will go slightly beyond that classical limit. We
will show that for a system which is only slightly non-commuting,
Problem~(\ref{def:dr}) is inside $\NP$ for sufficiently large $r$'s.
In other words, there cannot be a quantum PCP construction that
yields such slightly non-commuting systems with such large $r$'s.

This is the main theorem we wish to prove:
\begin{theorem}
  \label{thm:main} 
  
  For the set of QSAT systems that are defined on a $D$-regular
  graph using $d$-dimensional qudits, the
  following holds: if for every two projections,
  \begin{align}
    \label{eq:comm}
    \norm{[Q_i,Q_j]}\le \delta \ ,
  \end{align}
  and if the system is satisfiable (has a ground energy 0), then
  there exists an efficiently contractable tensor network with
  energy $\le \epsilon M$, where $0<\epsilon<1$ depends only on
  $d,D, \delta$, and $\epsilon\to 0$ as $\delta\to 0$. Consequently,
  Problem~\ref{def:dr} is inside NP when its projections satisfy
  \Eq{eq:comm} and $r>2\epsilon$.
\end{theorem}

The idea of the proof is very simple. Using the assumption that the
projections in $H$ nearly commute, we will find an auxiliary
\emph{commuting} system, $\hat{H}=\sum_{i=1}^M \hat{Q}_i$ such that
$\norm{Q_i-\hat{Q}_i}\le \epsilon/2$ for every $i$. By
\Ref{ref:Bra03}, this system has a ground state with an efficient
description, that can  thus be provided as a classical witness to the
$\NP$ verifier. The point is that this state is an $M\epsilon$ approximation of
the ground state, and thus can provide an $\NP$ witness for a
$rM \ge 2\epsilon M$ approximation of the ground energy.

\subsection{Further research}

It is interesting to see if the result of this note can be
strengthened. First functional dependence of $\epsilon$ on
$d,D,\delta$ is not given here, but it is probably not too difficult
to find one by generalizing the results of \Ref{ref:Pea79}. 

To show that Problem~\ref{def:dr} is inside NP for every $(d,D,r)$
we would like to show that for $\delta=1$, one can always find a
tensor-network that would yield an energy $\le \epsilon M$, for an
arbitrarily small, yet constant $\epsilon$. Such a result would be a
very strong indication against quantum PCP.

There are two natural directions that may help to prove such a
result. First, we note that the tensor-network that results from the
construction of Bravyi and Vyalyi is a very simple one. It is
essentially a depth-4 local quantum circuit. Using a classical
computer, we can in fact, contract similar networks with a
logarithmic depth. Can we find a perturbation theory for which the
depth-4 network is just the first order approximation? Then by going
to higher orders, we may systematically lower $\epsilon$ for a given
$d,D,\delta$.

Related to that, one may try to find a reduction of the system to a
system that is more commuting, perhaps by some sort of a
coarse-graining process. This again, might lead to a larger
$\epsilon$ for a given $d,D,\delta$. We now proceed to the proof of
theorem~\ref{thm:main}.

{~}

\section{Proof of theorem~\ref{thm:main}}

\noindent\textbf{Notation:}\\
We use the natural inner product on the space of matrices
\begin{align*}
  \inprod{A}{B} \EqDef \Tr(A^\dagger B) \ ,
\end{align*}
which leads to the Frobenius norm
\begin{align*}
  \norm{A}_F \EqDef \sqrt{\Tr(A^\dagger A)} \ .
\end{align*}

\begin{proof}  
  Using the $C^*$-algebra machinery of \Ref{ref:Bra03},
  we start with the following decomposition of every 2-local
  projection in $H$.
  \begin{lemma}
  \label{lem:prop}
    Let $Q$ be a 2-local projection on $\Qdit \otimes\Qdit$. Then
    one can write
    \begin{align}
      Q = \sum_{\alpha=1}^{d^2}  A_\alpha \otimes B_\alpha \ ,
    \end{align}
    with $A_\alpha$ and $B_\alpha$
    working locally on one particle, with the following properties:
    \begin{enumerate}
      \item $\{A_\alpha\}$ are orthogonal and bounded $\norm{A_\alpha}_F\le 1$.
        
      \item $\{B_\alpha\}$ are orthonormal.
        
      \item The algebra generated by $\{A_\alpha\}$ is close under
        conjugation, and the same applies for $\{B_\alpha\}$.
    \end{enumerate}
  \end{lemma}
  \begin{proof}
    We treat $Q$ as a vector in a bipartite Hilbert space of
    operators. Using the Schmidt decomposition in that space, 
    $Q=\sum_\alpha \lambda_\alpha\cdot A'_\alpha\otimes B_\alpha$,
    with $\{A'_\alpha\}$ and $\{B_\alpha\}$ orthonormal. Defining
    $A_\alpha \EqDef \lambda_\alpha A'_\alpha$ then proves 1) and 2). The
    third property follows from the Hermiticity of $Q$.
  \end{proof}
  
  The advantage of working in this representation is that a
  Frobenius distance between the $Q$'s easily translates into a
  Frobenius distance between the $A_\alpha$, and then the same
  applies for the usual operator norm because all norms are
  equivalent on finite dimensional spaces. Specifically, assume the
  above decomposition for adjacent projections $Q_1,Q_2$:
  \begin{align}
  \label{eq:expansion}
    Q_1 &= \sum_\alpha A^{(1)}_\alpha \otimes B^{(1)}_\alpha &
    Q_2 &= \sum_\beta  A^{(2)}_\beta \otimes B^{(2)}_\beta  \ ,
  \end{align}
  where $A^{(1)}_\alpha$ and $A^{(2)}_\beta$ operate on the same qudit, and
  $B^{(1)}_\alpha$ and $B^{(2)}_\beta$ on two other qudits. Then
  by the orthonormality of $B^{(i)}_\alpha$, 
  \begin{align}
    \norm{[Q_1,Q_2]}_F^2 &\EqDef
    \Tr\Big([Q_1,Q_2] ([Q_1,Q_2])^\dagger\Big) \nonumber \\
     &= \sum_{\alpha,\beta} 
       \Tr\Big([A^{(1)}_\alpha,A^{(2)}_\beta]
         ([A^{(1)}_\alpha,A^{(2)}_\beta])^\dagger \\
     &= \sum_{\alpha,\beta}
       \norm{[A^{(1)}_\alpha,A^{(2)}_\beta]}^2_F  \ .
  \end{align}
  Therefore, for every $\alpha,\beta$, we get
  $\norm{[A^{(1)}_\alpha,A^{(2)}_\beta]}_F \le \norm{[Q_1,Q_2]}_F$.
  
  The following lemma tells us that if the operators are slightly
  non-commuting we can find a set of near-by operators which fully
  commute. 
  
  \begin{lemma}
  \label{lem:main} There exists a function $\delta(\epsilon)$ with
    the limit $\lim_{\epsilon\to 0}\delta(\epsilon)\to 0$ such that
    the following holds. Every set of operators $\{ Q_i\} \ , i=1,
    \ldots, D$ that work on a given particle, and for which
    $\norm{[Q_i,Q_j]}\le \delta(\epsilon)$, can be replaced by the
    operators $\{ \hat{Q}_i\}$ with $\norm{Q_i - \hat{Q}_i} \le
    \epsilon/4$, that in addition satisfy the following properties
    \begin{enumerate}
      \item $[\hat{Q}_i, \hat{Q}_j] = 0$
      \item $\hat{Q}_i$ are Hermitian
      \item For any other term in the system, 
          $\norm{[\hat{Q}_i,P]}_F = \norm{[Q_i,P]}_F$, so that the
          system does not become less commuting at other places.
    \end{enumerate}
    Notice that $\delta(\epsilon)$ may depend on $d,D$.
  \end{lemma}
  
  \begin{proof}
    
    We use \Lem{lem:prop} to decompose
    \begin{align}
    \label{eq:Qi}
      Q_i = \sum_\alpha A^{(i)}_\alpha \otimes B^{(i)}_\alpha  \ .
    \end{align}
    Then by what was said above, for every $i\ne j$ and every
    $\alpha, \beta$, we have 
    $\norm{[A^{(i)}_\alpha,A^{(i)}_\beta]}_F \le \delta$. We now
    claim that the operators $\{ A^{(i)}_\alpha\}$ can be replaced
    by operators $\{ a^{(i)}_\alpha\}$ such that
    $\norm{A^{(i)}_\alpha - a^{(i)}_\alpha}_F \le c(d,D)\epsilon$, where
    $c(d,D)$ is some geometrical constant to be defined later. In addition, the
    operators $\{ a^{(i)}_\alpha\}$ enjoy the following properties
    \begin{enumerate}
      \item For every $i\ne j$ and $\alpha, \beta$,
        $[a^{(i)}_\alpha, a^{(j)}_\beta] = 0$.
      \item For every $i$, $\{a^{(i)}_\alpha\}$ are orthogonal with
      $\norm{a^{(i)}_\alpha}_F = \norm{A^{(i)}_\alpha}_F$.
      \item For every $i$, the algebra that is generated by
        $\{a^{(i)}_\alpha\}$ is close under conjugation.
    \end{enumerate}
    
    To show that these operators exist, together with the function
    $\delta(\epsilon)$, we use a neat argument from \Ref{ref:Hal76}
    (page 76)\footnote{This nice trick was first brought to my
    attention by Matthew Hastings \cite{ref:Has:commuting}}: Assume
    by contradiction that this is not true. Then there is a
    $\epsilon_0>0$ and a series of operators $\{A^{(i)}_\alpha(n)\}$
    and $\delta_n \to 0$ such that 
    \begin{align}
    \label{eq:limit0}
      \text{for every $i\ne j$ and $\alpha,\beta$}\qquad
      \norm{[A^{(i)}_\alpha(n), A^{(j)}_\beta(n)]}_F \le \delta_n \to 0 \ ,
    \end{align}
    and at the same time, for any set operators $\{a^{(i)}_\alpha\}$
    that fulfills the 3 requirements,  for every $n$, there is a at least one
    $a^{(i)}_\alpha$ such that
    \begin{align}
    \label{eq:contradiction}
      \norm{A^{(i)}_\alpha(n) - a^{(i)}_\alpha} \ge c(d,D)\epsilon_0 \ .
    \end{align}
    
    But this is obviously wrong, since as we are working in a
    compact space (we work with bounded operators in a
    finite-dimensional Hilbert space), the series
    $\{A^{(i)}_\alpha(n)\}$ must have at least one limit point.
    Denote this point by $\{a^{(i)}_\alpha\}$. Then it is easy to
    see that by \Eq{eq:limit0} that property 1) must hold, and by
    \Lem{lem:prop} and the fact that they are a limit point of
    $A^{(i)}_\alpha(n)$, properties 2) and 3) must hold.
    
    Having the operators $\{ a^{(i)}_\alpha\}$ at our hands, we
    construct $\{\hat{Q}_i\}$ by simply replacing $A^{(i)}_\alpha
    \mapsto a^{(i)}_\alpha$ in \Eq{eq:Qi}. By definition, the new
    operators commute between themselves, proving property 1). In
    addition, by using the equivalence of the operator norm and the
    Frobenius norm, it is an easy exercises to choose $c(d,D)$ such
    that $\norm{\hat{Q}_i-Q_i}\le \epsilon/4$. To prove 3), let us consider
    $Q \in \{Q_i\}$ and its replacement $\hat{Q}$, as well as a
    third operators $P$ that intersects with $Q$ but that does not
    belong to $\{Q_i\}$. We use \Lem{lem:prop} to write
    $Q=\sum_\alpha A_\alpha \otimes B_\alpha$, $\hat{Q}=\sum_\alpha
    a_\alpha \otimes B_\alpha$, $P=\sum_\beta C_\beta \otimes
    D_\beta$, where it is assumed that $B_\alpha$ and $C_\beta$ work
    on the same particle. Then using the orthogonality of $A_\alpha,
    a_\alpha$ and $D_\beta$, we conclude that
    \begin{align*}
      \norm{[Q,P]}^2_F &= \sum_{\alpha,\beta}
      \norm{[B_\alpha,C_\beta]}^2_F \cdot \norm{A_\alpha}^2_F
       \cdot \norm{D_\beta}^2_F \ , \\
      \norm{[\hat{Q},P]}^2_F &= \sum_{\alpha,\beta}
      \norm{[B_\alpha,C_\beta]}^2_F \cdot \norm{a_\alpha}^2_F
       \cdot \norm{D_\beta}^2_F \ ,
    \end{align*} 
    and so 3) follows from the fact that $\norm{A_\alpha}_F =
    \norm{a_\alpha}_F$.
    
    Finally, we show how the $\hat{Q}_i$ can be made Hermitian. This
    essentially follows from the fact that for every $i$,
    $\{a^{(i)}_\alpha\}$ are closed under conjugation. Using this
    property, it is easy to see that not only $[\hat{Q}_i,
    \hat{Q}_j]=0$, but that also $[\hat{Q}^\dagger_i, \hat{Q}_j]=0$.
    So by replacing $\hat{Q}_i \mapsto \frac{1}{2}(\hat{Q}_i +
    \hat{Q}_i^\dagger)$ we make the operators Hermitian while
    maintaining the other useful properties.
    
  \end{proof}

  We can now finish the proof of Theorem~\ref{thm:main}. Using the
  last lemma, we construct a new 2-local Hamiltonian by sequentially
  going over all particles and replacing $Q_i \mapsto \hat{Q}_i$. We
  obtain a new system $\hat{H}=\sum_i \hat{Q}_i$, which is commuting
  and the distance between two corresponding operators $Q_i$ and
  $\hat{Q}_i$ is 
  \begin{align}
    \label{eq:dist}
    \norm{ Q_i - \hat{Q}_i} \le \epsilon/2 \ .
  \end{align}
  (This is because $Q_i$ works on 2 particles, and so it undergoes
  two replacements, giving us a total error of
  $\epsilon/4+\epsilon/4=\epsilon/2$). Let
  $\ket{\psi_0}$ be the ground state of the original system. By
  assumption, $\bra{\psi_0}H\ket{\psi_0}=0$. Then by \Eq{eq:dist},
  $\bra{\psi_0}\hat{H}\ket{\psi_0} \le M\epsilon/2$, hence the ground
  energy of $\hat{H}$ is at most $M\epsilon/2$. By \Ref{ref:Bra03},
  $\hat{H}$ has a ground state $\ket{\psi_0'}$ that can be written
  as an efficient tensor-network. By using \Eq{eq:dist} once
  more, we see that $\bra{\psi_0'}H\ket{\psi'_0} \le \epsilon M$, as
  required.
\end{proof}

{~}

\noindent{\bf Acknowledgments}\\ 
I wish to thank Dorit Aharonov,
Matt Hastings, Lior Eldar, , Zeph Landau, Tobias Osborne, Umesh
Vazirani, for inspiring discussions about the quantum PCP and to
Dorit Aharonov for her great help on the manuscript. 

I gratefully acknowledge support by Julia Kempe's ERC Starting
Researcher Grant QUCO and ISF 759/07.

\bibliographystyle{alpha}
\bibliography{QC}

\begin{thebibliography}{AALV09}

\bibitem[AALV09]{ref:Aha09}
Dorit Aharonov, Itai Arad, Zeph Landau, and Umesh Vazirani.
\newblock The detectability lemma and quantum gap amplification.
\newblock In {\em STOC '09: Proceedings of the 41st annual ACM symposium on
  Theory of computing, \texttt{arXiv:0811.3412}}, pages 417--426, New York, NY,
  USA, 2009. ACM.

\bibitem[{Bra}06]{ref:Bra06}
S.~{Bravyi}.
\newblock {Efficient algorithm for a quantum analogue of 2-SAT}.
\newblock {\em arXiv:quant-ph/0602108}, 2006.

\bibitem[BV03]{ref:Bra03}
S.~Bravyi and M.~Vyalyi.
\newblock {Commutative version of the k-local Hamiltonian problem and common
  eigenspace problem}.
\newblock {\em arXiv:quant-ph/0308021}, 2003.

\bibitem[Din07]{ref:Din07}
Irit Dinur.
\newblock The pcp theorem by gap amplification.
\newblock {\em J. ACM}, 54(3):12, 2007.

\bibitem[Hal76]{ref:Hal76}
PR~Halmos.
\newblock {Some unsolved problems of unknown depth about operators on Hilbert
  space}.
\newblock In {\em Proc. R. Soc. Edinb., Sect. A}, volume~76, pages 67--76,
  1976.

\bibitem[Has08]{ref:Has:nopcp}
M.~B. Hastings.
\newblock personal communication, 2008.

\bibitem[Has09]{ref:Has:commuting}
M.~B. Hastings.
\newblock personal communication, 2009.

\bibitem[PH10]{ref:Pou10}
D.~{Poulin} and M.~B. {Hastings}.
\newblock {Markov entropy decomposition: a variational dual for quantum belief
  propagation}.
\newblock {\em arXiv:1012.2050}, 2010.

\bibitem[PS79]{ref:Pea79}
Carl Pearcy and Allen Shields.
\newblock Almost commuting matrices.
\newblock {\em Journal of Functional Analysis}, 33(3):332 -- 338, 1979.

\end{thebibliography}

\end{document}